\documentclass[journal]{IEEEtran}

 \ifCLASSINFOpdf
   \usepackage[pdftex]{graphicx}

\else

\fi

\usepackage{tikz}
\usetikzlibrary{shapes,arrows,fit,positioning}
\usepackage[subtle]{savetrees}
\usepackage{xcolor, amsmath}
\usepackage{amsfonts}
\usepackage{amsthm}
\usepackage{mathrsfs}
\usepackage{cleveref}
\usepackage[linesnumbered,ruled,vlined]{algorithm2e}
\usepackage{cite}

\newtheorem{theorem}{Theorem}
\newtheorem{remark}{Remark}

\begin{document}

\title{Data-driven System Interconnections and a Novel Data-enabled Internal Model Control}

\author{Yasaman~Pedari,~\IEEEmembership{Student Member,~IEEE,}
        Jaeho Lee,~\IEEEmembership{Student Member,~IEEE,}
        Yongsoon Eun,~\IEEEmembership{Senior Member,~IEEE,} and~Hamid~Ossareh,~\IEEEmembership{Senior Member,~IEEE}
\thanks{This work was supported by NASA under agreement VT-80NSSC20M0213 and by National Science Foundation award CMMI-2238424.}
\thanks{Y. Pedari and H. Ossareh are with the Dept. of Elec. \& Biomed. Eng., Univ. of Vermont, Burlington, VT, USA ({\tt\small ypedari, hossareh@uvm.edu}). J. Lee and Y. Eun are with the Dept. of Elec. Eng. \& Comp. Sci., DGIST, Daegu, South Korea ({\tt\small jaeho.lee,yeun@dgist.ac.kr}).}%
}

\maketitle

\begin{abstract}

Over the past two decades, there has been a growing interest in control systems research to transition from model-based methods to data-driven approaches. In this study, we aim to bridge a divide between conventional model-based control and emerging data-driven paradigms grounded in Willems' ``fundamental lemma". Specifically, we study how input/output data from two separate systems can be manipulated to represent the behavior of interconnected systems, either connected in series or through feedback. Using these results, this paper introduces the Internal Behavior Control (IBC), a new control strategy based on the well-known Internal Model Control (IMC) but viewed under the lens of Behavioral System Theory. Similar to IMC, the IBC is easy to tune and results in perfect tracking and disturbance rejection but, unlike IMC, does not require a parametric model of the dynamics. We present two approaches for IBC implementation: a component-by-component one and a unified one. We compare the two approaches in terms of filter design, computations, and memory requirements.

\end{abstract}


\section{Introduction}
Recently, there has been a growing interest in data-driven control methods, mainly due to a surge in available large datasets. 
These methods potentially offer seamless end-to-end and flexible design, particularly for systems that are difficult to model precisely.
Following the trend towards data-driven control, the work by Willems and coauthors on behavioral system theory (BST) stands out\cite{willems1997introduction}. They demonstrated that a sufficiently exciting trajectory of a Linear Time-Invariant (LTI) system {\color{black}} can be used to represent all its potential trajectories. This concept termed as the ``fundamental lemma", highlights the power of using data for learning and control. 

In this study, we aim to bridge a divide between conventional model-based control and emerging data-driven paradigms grounded in Willems' fundamental lemma. Specifically, we present how noise-free\footnote{Handling noisy offline data within the context of BST is an active area of research and is a topic for future work.} input/output data from multiple systems can be used to represent the dynamics of an interconnected system in the BST framework. The existing literature (e.g., \cite{willems2007behavioral, willems1991paradigms,willems1989models}) has also studied the interconnection problem {\color{black} in the context of BST}, but not from an algorithmic standpoint, rather from a purely theoretical standpoint as constraints on trajectories or latent variables, or through parametric  (i.e., state-space) representations of the dynamics. 
Applications of \cite{willems2007behavioral, willems1991paradigms,willems1989models} in controls have been reported in \cite{maupong2017data,maupong2016data,trentelman1996control,willems1997interconnections}.
Building on these foundations, the first contribution of this paper is an
 algorithmic approach to convert individual system trajectories collected in isolation into a data-driven representation of the interconnected system, either in series or in feedback. 

Using these data-driven interconnections, we propose a new control framework called the Internal Behavior Control (IBC), which can be viewed as a data-driven version of the classical Internal Model Control (IMC). The IMC uses a model of the plant, as well as a causal approximate inverse, inside the controller \cite{garcia1982internal}. Here, we replace the model and the inverse with data-enabled predictors. As we show, IBC inherits the desirable properties of the model-based IMC: integral control, ease of tuning, and ability to naturally handle Hammerstein-Weiner systems (e.g., actuator saturation). 
We offer two {\color{black} approaches for implementing} the proposed IBC. The first implementation is component-by-component, where separate data-enabled predictors are leveraged for the forward and inverse models. The second implementation combines the data-driven forward and inverse predictors using the system interconnection ideas discussed earlier. This effectively results in a single data-enabled controller, which is more compact and more memory efficient than the first implementation. The focus of this paper will be on single-input single-output {\color{black}(SISO)}, stable, minimum-phase systems, though extension to the multi-input multi-output case is not difficult. 

Note that other model-based control strategies have also been recently viewed under the lens of BST. Successful examples include the DeePC algorithm \cite{coulson2019data}, which stands parallel to the Model Predictive Control (MPC), and the design of optimal Linear Quadratic Gaussian controllers directly from the subspace matrices in \cite{favoreel1999model}. These controllers may lack interpretability and may be hard to tune. The proposed IBC does not suffer from these challenges.


In summary, this work mainly offers two contributions: {\it i}) 
methods and algorithms for data-driven system interconnections; and {\it ii}) the IBC, which is a data-driven control framework inspired by the classical IMC. To the best of our knowledge, with the exception of the work \cite{rueda2020data} that studies a similar problem for control of second-order Volterra systems, this is the first study of the IMC controller in the context of BST. 

This paper is organized as follows. Sec. II provides an overview of BST and IMC.  Data-driven interconnections and the IBC are presented in Sections \ref{sec:inter} and  \ref{sec:IBC}, respectively. The conclusions and future work are provided in Section V.

The notation in this work is as follows. Given a system $G$, the integers $n({G})$ and $L({G})$ denote the order and $L$-delay (i.e., the relative degree in the SISO case), respectively. For a scalar signal $v{(t)}\in\mathbb{R}$, 
the vector of values from time $t_1$ to $t_2$ is denoted by
$\mathbf{v}_{[t_1,t_2]} = [v(t_1), \ldots, v(t_2)]^\top$.  When there is no risk of confusion, we omit the subscript and denote the vector by boldface letters, i.e., $\mathbf{v}$.
Given two vectors, $\mathbf{x}$ and $\mathbf{y}$, we define $\mathrm{col}(\mathbf{x},\mathbf{y}) = \begin{bmatrix}
    \mathbf{x}^T & \mathbf{y}^T
\end{bmatrix}^T$.
For the vector $\mathbf{v}_{[t_1,t_2]}$, the Hankel matrix of order $T$ is defined as \par
{\small {\begin{equation*}
    \mathscr{H}_{T}(\mathbf{v}_{[t_1,t_2]}) = \begin{bmatrix}
        v{(t_1)} & v{(t_1+1)} & \dots & v{(t_2 - T + 1)} \\
        v{(t_1+1)} & v{(t_1+2)} & \dots & v{(t_2 - T)} \\
        \vdots & \vdots & \ddots & \vdots \\
        v{(t_1+T-1)} & v{(t_1+T)} & \dots & v{(t_2)}
    \end{bmatrix}.
\end{equation*}
}}

\section{Overview of BST and IMC}

\subsection{Overview of BST and Data-enabled predictions}\label{sub:overview BST}
This section provides a review of BST, for details see \cite{MARKOVSKY202142}. Consider a {\color{black}SISO} LTI system $G$ with input $u \in \mathbb{R}$ and output $y \in \mathbb{R}$. 
For any positive integer $T$, 
 we define the ``restricted behavior" of $G$, denoted by $\mathscr{B}|_T(G)$, as the set of all $T$-long trajectories $\textbf{w}=\mathrm{col}(\textbf{u},\textbf{y})\in \mathbb{R}^{2T}$, where 
 $\textbf{y} \in \mathbb{R}^T$ is the output trajectory corresponding to $\textbf{u} \in \mathbb{R}^T$ from some initial condition. 

Consider an offline-collected single $T_d$-long trajectory, $\textbf{w}^d = \mathrm{col}(\textbf{u}^d,\textbf{y}^d)\in \mathbb{R}^{2T_d}$ of system $G$. Let $T_p$ be a positive integer to be later selected. 
We construct the ``data matrix", $\mathcal{H}$, 
as
\begin{equation}\label{eq: hankel}
    \mathcal{H} = \begin{bmatrix}
    \mathcal{H}_u \\ \mathcal{H}_y
\end{bmatrix}\in \mathbb{R}^{2(T_p+1)\times (T_d-T_p)}
\end{equation} where 
\begin{equation*}
    \mathcal{H}_u = \mathscr{H}_{T_p+1}(\textbf{u}^d) =\begin{bmatrix}
    U_p \\ U_f
\end{bmatrix},\quad \, \mathcal{H}_y = \mathscr{H}_{T_p+1}(\textbf{y}^d) = \begin{bmatrix}
    Y_p \\ Y_f
\end{bmatrix},\end{equation*} with the partitioning $U_p, Y_p \in \mathbb{R}^{T_p\times(T_d-T_p)}$ and $U_f, Y_f \in \mathbb{R}^{1\times(T_d-T_p)}$. 
 If $T_p \geq n(G)$ and $\mathcal{H}$ satisfies the ``low-rank condition'',
\begin{equation}\label{eq:low rank condition}
\text{rank}(\mathcal{H}) = T_p+1+n(G),
\end{equation}
then, any $(T_p+1)$-long trajectory of the system will belong to the column space of $\mathcal{H}$\footnote{This result has become known as the ``Generalized Persistency of Excitation Condition", see \cite[Corollary 21]{markovsky2022identifiability}. The condition formally involves the ``lag" of the system, which in the SISO case is equal to the order, $n$.}. This condition allows us to use the data matrix for the purpose of output prediction. Specifically, given the recent $T_{p}$ samples of the input/output trajectory ($\mathbf{u_{ini}}\in \mathbb{R}^{T_{p}}$ and $\mathbf{y_{ini}}\in \mathbb{R}^{T_{p}}$), and a ``future" input ${u_\text{pred}}\in \mathbb{R}$, the  \textbf{unique} future output, ${y_\text{pred}}\in \mathbb{R}$, can be predicted by solving:\par
{\small
\begin{equation}\label{eq:predictor}
\mathcal{H}{\mathbf{g}} = \left[ \begin{array}{c}
\mathbf{u_{ini}}\\
{u_\text{pred}}\\
\mathbf{y_{ini}}\\
{y_\text{pred}}\end{array}\right] 
\end{equation}}\par
\noindent for the unknowns $\mathbf{g}$ and ${y_\text{pred}}$. 
The solution for vector $\mathbf{g}$ is generally not unique. Often, the minimum-norm solution is employed:
\begin{equation}\label{eq:g*}
\mathbf{g}^* = \left[\begin{array}{c}U_p\\U_f\\Y_p \end{array}\right]^\dagger\left[\begin{array}{c}\mathbf{u_{ini}}\\{u_\mathrm{pred}}\\\mathbf{y_{ini}} \end{array}\right]
\end{equation}
and so
\begin{equation}\label{eq:ypred}
{y_\mathrm{pred}} = Y_f \mathbf{g}^*.
\end{equation}

\begin{remark}
    Two assumptions in this review are not always necessary: First, the data matrix does not require a single long trajectory; it can be formed from independent trajectories; for details see \cite{MARKOVSKY202142}. Second, our predictors are single-step predictors (i.e., $u_{\text{pred}}, y_{\text{pred}} \in \mathbb{R}$). However, longer prediction horizons may be used provided the data matrix satisfies the low-rank condition.
\end{remark}

{\color{black}\begin{remark}
Because of the use of a single step prediction (i.e., $y_\mathrm{pred} \in \mathbb{R}$), the solution to \eqref{eq:g*} and \eqref{eq:ypred} is equivalent to the prediction obtained from a certainty-equivalent ARX model identified from the offline data. For details, please see \cite{dorfler2022bridging}. 

    
\end{remark}
}

\subsection{Data-Enabled prediction of System Inverse}\label{sub: dd pred}

This section provides a review of data-enabled predictions for the system inverse, as required for the IBC approach; for details see \cite{eun2023data}. 
Given an output trajectory, the goal is to compute the associated input using the offline data.  
The ideas are similar to the previous section,  except that we predict the input instead of the output and $\mathcal{H}_y$ is now of a different order. More specifically, the data matrix for the inverse system is modified to
\begin{equation}\label{eq: hankel2}
    \mathcal{H}_{\text{inv}} = \begin{bmatrix}
    \mathcal{H}_u \\ \mathcal{H}_{y,\text{inv}}
\end{bmatrix}\in \mathbb{R}^{(2T_p+2+L(G))\times (T_d-T_p)}
\end{equation} where $\mathcal{H}_u$ is as in \eqref{eq: hankel} and  $\mathcal{H}_{y,\text{inv}}$ is defined as:
\begin{equation*}
    \mathcal{H}_{y,\text{inv}} = \mathscr{H}_{T_p+L(G)+1}(\textbf{y}^d) = \begin{bmatrix}
    Y_p \\ Y_{fL}
\end{bmatrix},\end{equation*} with the partitioning $Y_p \in \mathbb{R}^{T_p\times(T_d-T_p)}$ and $Y_{fL} \in \mathbb{R}^{(1+L(G))\times(T_d-T_p)}$. Here, we have considered different lengths for the offline data: $T_d$ for $\mathbf{u}^d$  and $T_d+L(G)$ for $\mathbf{y}^d$. 
 The low-rank condition for the inverse case remains the same as \eqref{eq:low rank condition}. If this condition is satisfied, $\mathcal{H}_{\text{inv}}$ can be used to compute the input. To see this, at every timestep $t$, given the recent $T_p + 1 + L(G)$ samples of the output and the same samples of the input without the last $L(G)$ elements, we can uniquely compute the input $L(G)$ timesteps in the past. Mathematically, given 
$$
\mathbf{u}_{\text{ini}}^{\text{inv}} = \mathbf{u}_{[t-T_p-L(G), t-L(G)-1]}, 
    $$
    $$
    \mathbf{y}_{\text{ini}}^{\text{inv}} = \mathbf{y}_{[t-T_p-L(G), t-L(G)-1]}, \,\,\, \quad
    \mathbf{y_{\mathrm{pred}}^{\mathrm{inv}}} = \mathbf{y}_{[t-L(G),t]},
$$
${u_{\mathrm{pred}}^{\mathrm{inv}}= u(t-L(G))\in \mathbb{R}}$ can be obtained by solving:
\begin{equation}\label{eq:inv}
    \begin{bmatrix}
        U_p\\U_f\\Y_p\\Y_{fL}
    \end{bmatrix}\mathbf{g}^\mathrm{inv} = \begin{bmatrix}
        {\mathbf{u}_{{\mathrm{ini}}}^{\mathrm{inv}}}\\
        {u_{\mathrm{pred}}^{\mathrm{inv}}}\\{\mathbf{y}_{{\mathrm{ini}}}^{\mathrm{inv}}}\\
    {\mathbf{y}_{{\mathrm{pred}}}^{\mathrm{inv}}}
    \end{bmatrix}
\end{equation} for the unknowns $\mathbf{g}^\mathrm{inv}\in \mathbb{R}^{(T_d-T_p)\times(2T_p+L(G)+2) }$ and $u_{\mathrm{pred}}^{\mathrm{inv}}$.  
As in the case of forward prediction, 
the minimum-norm solution to $\mathbf{g}^\mathrm{inv}$ is often employed, which leads to:
\begin{equation}\label{eq:inv sol}
{u_{\mathrm{pred}}^{\mathrm{inv}}} =U_f\left[\begin{array}{c}U_p\\Y_p\\Y_{fL} \end{array}\right]^\dagger\left[\begin{array}{c}{\mathbf{u}_{{\mathrm{ini}}}^{\mathrm{inv}}}\\{\mathbf{y}_{{\mathrm{ini}}}^{\mathrm{inv}}}\\
{\mathbf{y}_{{\mathrm{pred}}}^{\mathrm{inv}}} \end{array}\right].
\end{equation}


\subsection{Overview of IMC}
A block diagram of the IMC structure is shown in Fig. \ref{fig:IMCBD}, where ${G}$ represents the system, $\hat{{G}}$ is a model of the system, $\hat{{G}}^{-1}$ is the ideal inverse of the model (to be revisited), and ${F}$ is a filter to be defined. The goal is to select ${F}$ to achieve a certain performance objective, e.g., to minimize the $H_2$-norm of the tracking error for step setpoints. For open-loop-stable minimum-phase systems, this $H_2$-norm is theoretically minimized by setting ${F}=1$. 
In practice, ${\hat{G}}^{-1}$ is improper and cannot be implemented. Moreover, even if it could be implemented, it would render the feedback loop sensitive to uncertainties. This requires the augmentation of ${\hat{G}}^{-1}$ with a low-pass filter whose order is at least the relative degree of ${\hat{G}}$. If a description of $\hat{{G}}$'s uncertainties is available, the time constant of the filter is selected to ensure robust performance and stability \cite{morari1987robust, rivera1986internal}. Otherwise, the time constant is treated as a tuning parameter to directly trade-off performance and robustness, rendering tuning an easy task. Note that input or output nonlinearities (i.e., Hammerstein-Wiener models \cite{aguirre2005interpretation}) can be naturally compensated for by including them in the model path. In our study, we aim to substitute $\hat{{G}}$ and ${\hat{G}}^{-1}$ with a data-enabled predictor while preserving IMC's key features, namely perfect steady-state tracking and disturbance rejection, ease of tuning, and minimization of $H_2$ norm for step commands.

\begin{figure}
    \centering
    \begin{tikzpicture}[auto, node distance=2cm,>=latex',
  sum/.style={
    draw, 
    circle, 
    inner sep=2pt
  }]
    \node [draw, rectangle, minimum width=1.5cm] (G1) { \( \hat{G}^{-1} \) };
    \node [draw, rectangle, right of=G1 , node distance=2cm, minimum width=1.5cm] (G2) { \( F \) }; 
    \node [draw, rectangle, right of=G2 , node distance=2cm, minimum width=1.5cm] (G3) { \( G \) }; 
    \node [draw, rectangle, below of=G3 , node distance=1cm, minimum width=1.5cm] (G4) { \( \hat{G} \) }; 
    \node[sum,  left=0.5cm of G1] (sum1) {};
    \node[sum,  right=0.5cm of G4] (sum2) {};

    \draw[->]   ([xshift=-1cm]sum1.east)-- node[name=s1, midway]{ \( r \) }(sum1);
    \draw[->] (G3) -- node[name=y, midway] { \( y \) } ([xshift=1cm]G3.east);
    
    \draw[->] (sum1.east) -- node[name=s1, midway]{ \( s_1 \) }(G1.west);
    \draw[->] (G1.east) -- node[name=s2, midway]{ \( s_2 \) }(G2.west);
    \draw[->] node[name=u, midway]{}(G2.east) -- (G3.west);
    \draw[->] (G2.east) -- node[name=u, midway]{ \( u \) }(G3.west);
    \draw[->] (G4.east) -- node[name=yhat, midway]{ \( \hat{y} \) }(sum2);
    \draw[->] (G2.east) -- ++(0.2cm,0) -- ++(0,-1cm) --  (G4.west);
    \draw[->] (G3.east) -- ++(0.6cm,0) -- (sum2.north);
    \draw[->] (sum2.south) -- ++(0,-0.5cm) -- ++ (-6.75cm,0cm)--node[name=s1, midway]{ \( e \) }(sum1.south);
    \node[anchor=north east] at ([xshift=0.1cm,yshift=-0.1cm]sum1.west) {-}; 
    \node[anchor=north east] at ([xshift=0.1cm,yshift=-0.1cm]sum2.west) {-};
    \node[anchor=north east] at ([xshift=0.5cm,yshift=0.5cm]sum2.west) {+};

\end{tikzpicture}

    \caption{Block diagram of the IMC structure}
    \label{fig:IMCBD}
\end{figure}
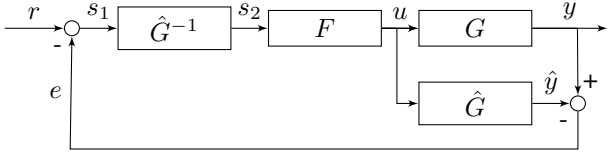

\section{Data-driven System Interconnections}\label{sec:inter}
In this section, we present how to create a data-driven representation of the interconnection of two system
-- whether they are in series or in feedback -- using available input/output data from each individual system. 
As we show in the next section, this allows us to view 
traditional model-based control techniques, which frequently depend on system block interconnections, under the lens of data-driven predictors. 


As discussed in \cite{willems1997interconnections}, in order to interconnect two systems, the output of one system must match the input of the other. However, in general, the offline input-output data may be collected in isolation, leading to mismatches. To tackle this issue, we introduce the following function, the specific form of which will be defined later: 
\begin{equation*}
   \mathbf{y}^* := \mathcal{Z}(\mathbf{u}^*,\, \mathbf{u}^d,\, \mathbf{y}^d): \mathbb{R}^{T}\times \mathbb{R}^{T_d} \times \mathbb{R}^{T_d} \to \mathbb{R}^{T}.
\end{equation*}
In words, given an input sequence $\mathbf{u}^*$, this function predicts the output sequence, $\mathbf{y}^*$, of the system starting from zero initial conditions, while ensuring that $\mathrm{col}(\mathbf{u^*},\mathbf{y^*})$ is an admissible trajectory, i.e., compatible with  $\mathbf{w}^d = \mathrm{col}(\mathbf{u}^d,\mathbf{y}^d)$.
 To this end, we introduce a new input trajectory $\mathbf{u}^\mathrm{mod}$, which is identical to $\mathbf{u^*}$ except it is padded with $T_p$ zeros at the beginning:
 \begin{equation*}
    u^{\text{mod}}{(t)} =  \begin{cases} 
    0 & \text{if } 0 \leq t < T_{p}, \\
 {u^*}(t-T_p)& \text{if } T_{p} \leq t \leq T+T_p. 
\end{cases}
\end{equation*} 
 We then define $\mathcal{Z}$ recursively, where 
$\mathbf{y}^*$ is computed element by element.
Specifically, using \eqref{eq:g*} and \eqref{eq:ypred}, ${y}^*(t)$ is given by:


\begin{equation}
y^*(t)={y^\mathrm{mod}}(t+T_p),
\end{equation}
where $y^{\text{mod}}{(t)}$ is computed recursively as
\begin{equation*}
    y^{\text{mod}}{(t)} =  \begin{cases} 
    0 & \text{if } 0 \leq t < T_{p}, \\
Y_f \left[\begin{array}{c}U_p\\U_f\\Y_p \end{array}\right]^\dagger\left[\begin{array}{c}\mathbf{u}^{\text{mod}}_{[t-T_{p}, t-1]}\\u^{\text{mod}}{(t)}\\\mathbf{y}^{\text{mod}}_{[t-T_{p}, t-1]} \end{array}\right]& \text{if } T_{p} \leq t \leq T+T_p.  \\
\end{cases}
\end{equation*} 
Here, the Hankel matrices $U_p$, $U_f$, $Y_p$, and $Y_f$ are constructed using $\mathbf{u}^d$ and $\mathbf{y}^d$, see \eqref{eq: hankel}. 
Note that $\mathcal{Z}$ is only defined when the data matrix of $\mathbf{w}^d$ has rank $T_p+1+n(G)$ for a selected $T_p \geq n(G)$. 



The above ideas can be used to ``remove" the effect of initial conditions from the offline output $\mathbf{y}^d$, as shown in the theorem below. The proof is straightforward and omitted.  

\begin{theorem}

Consider a SISO LTI system $G$ and let $T_p\geq n(G)$. Suppose a $T_d$-long trajectory $\mathbf{w}^d = \mathrm{col}(\mathbf{u}^d,\mathbf{y}^d)$ of $G$, whose data matrix has rank $T_p+1+n(G)$, is given. The trajectory $\mathbf{w}_0 = \mathrm{col}(\mathbf{u}^d,\mathcal{Z}(\mathbf{u}^d,\, \mathbf{u}^d,\, \mathbf{y}^d))$ is a trajectory of $G$ starting from zero initial conditions. Furthermore, if $\mathbf{w}^d$ is collected starting from zero initial conditions, then $\mathbf{w}_0 = \mathbf{w}^d$. 

\end{theorem} 


We now use $\mathcal{Z}$ to study the data-driven interconnection of two systems in series (see Fig. \ref{fig:series}), where their corresponding offline data may be of different lengths and may be collected in isolation. 
The following Theorem provides a mechanism to construct an admissible trajectory of the series interconnection directly from the offline data of the individual systems. Such trajectory can then be used for data-driven prediction of the series interconnection using the previously described methods. 

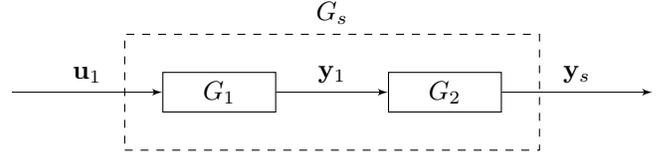
\begin{figure}
    \centering
\begin{tikzpicture}[auto, node distance=2cm,>=latex']
    \node [draw, rectangle, minimum width=1.5cm] (G1) { \( G_1 \) };
    \node [draw, rectangle, right of=G1, node distance=3cm, minimum width=1.5cm] (G2) { \( G_2 \) }; 
    \draw[->] (G1) -- node[name=y1, midway] { \( \mathbf{y}_1 \) } (G2);
    \draw[->] ([xshift=-2cm]G1.west) -- node[name=u1, midway] { \( \mathbf{u}_1 \) } (G1);
    \draw[->] (G2) -- node[name=ys, midway] { \( \mathbf{y}_s \) } ([xshift=2cm]G2.east);
    \node [draw, dashed, fit=(G1)(G2), inner sep=0.5cm, label=above:{ \( G_s \) }] {};
\end{tikzpicture}
    \caption{Series interconnection given data from individual systems. 
    In general, since ${\mathbf{y}}_1\neq {\mathbf{u}}_2$, the interconnection is not achievable. In the figure, the successful interconnection is shown with $G_2$ taking ${\mathbf{y}}_1$ as input and generating $\mathbf{y}_s$ as output.}
    \vspace{-0.2cm}
    \label{fig:series}
\end{figure}


\begin{theorem}\label{th: series}
    Consider $\mathbf{w}_1=\mathrm{col}(\mathbf{u}_1,\mathbf{y}_1)\in \mathscr{B}|_{T_1}(G_1)$ and $\mathbf{w}_2=\mathrm{col}(\mathbf{u}_2,\mathbf{y}_2)\in \mathscr{B}|_{T_2}(G_2)$
    as $T_1$ and $T_2$ long trajectories of SISO LTI systems $G_1$ and $G_2$, respectively. {\color{black}If the data matrix of trajectory  $\mathbf{w}_2$ has rank $T_p+1+n(G_2)$ for $T_p \geq n(G_2)$}, then the trajectory define by 
    \begin{equation*}
        \mathbf{w}_s = \begin{bmatrix}
        \mathbf{u}_1\\\mathbf{y}_s
    \end{bmatrix}\text{, }\, \mathbf{y}_s = \mathcal{Z}({\mathbf{y}}_1, {\mathbf{u}}_2, \, {\mathbf{y}}_2)
    \end{equation*}
    satisfies $\mathbf{w}_s \in \mathscr{B}|_{T_1}(G_s)$, where $G_s$ denotes the LTI system that results from connecting systems $G_1$ and $G_2$ in series.
\end{theorem}

\begin{proof}
To interconnect $G_1$ and $G_2$ in series, 
$G_1$'s output must match $G_2$'s input. However, the two trajectories, $\mathbf{w}_1$ and $\mathbf{w}_2$ generally do not satisfy this condition as they may have been collected in isolation. 
To force the output of $G_1$ to be equal to the input of $G_2$, we use $\mathcal{Z}$ to generate $G_2$'s output, $y_s$, given $G_1$'s output, $y_1$, as input. 
The trajectory $\mathrm{col}(\mathbf{u}_1,\mathbf{y}_s)$ would then represent a valid trajectory of the two systems in series, where the second system starts from zero initial conditions. An illustration of the proof is given in Fig. \ref{fig:series}.
\end{proof}

We now study a similar result for the feedback interconnection of two systems (see Fig. \ref{fig:feedback}).  

\begin{theorem}\label{th: feedback}
    Consider $\mathbf{w}_1=\mathrm{col}(\mathbf{u}_1,\mathbf{y}_1)\in \mathscr{B}|_{T_1}(G_1) \text{\, and \,} \mathbf{w}_2=\mathrm{col}(\mathbf{u}_2,\mathbf{y}_2)\in \mathscr{B}|_{T_2}(G_2)$
   as $T_1$ and $T_2$ long trajectories of SISO LTI systems $G_1$ and $G_2$, respectively. {\color{black}If the data matrix of $\mathbf{w}_2$ has rank $T_p+1+n(G_2)$ for $T_p \geq n(G_2)$}, then the trajectory defined by
       \begin{equation*}
         \mathbf{w}_f =\mathrm{col}(
        \mathbf{u}_1-\check{\mathbf{y}}_2, \, {\mathbf{y}}_1), \quad \, \check{\mathbf{y}}_2 =  \mathcal{Z}({\mathbf{y}}_1,\mathbf{u}_2,\mathbf{y}_2)
    \end{equation*}
   satisfies {\color{black}$\mathbf{w}_f \in \mathscr{B}|_{T_1}(G_f)$}, where $G_f$ denotes the interconnected LTI system that results from connecting systems $G_1$ and $G_2$ through positive feedback (shown in Fig. \ref{fig:feedback}).
\end{theorem}

\begin{figure}
    \centering
    
\begin{tikzpicture}[auto, node distance=2cm,>=latex',
  sum/.style={
    draw, 
    circle, 
    inner sep=2pt
  }]
    \node [draw, rectangle, minimum width=1.5cm] (G1) { \( G_1 \) };
    \node [draw, rectangle, below of=G1 , node distance=1cm, minimum width=1.5cm] (G2) { \( G_2 \) }; 
    \node[sum,  left=1cm of G1] (sum) {};
    \draw[->] (G1.east) -- ++(1cm,0) -- node[pos=0.1, left,name=p]{} ++(0,-1cm) --  (G2.east);

    \draw[->] (G2.west) -- ++  (-1.1cm,0)  -- (sum.south) node[pos=0.4,right] { \(  \check{\mathbf{y}}_2 \) };
    \draw[->] (sum.east) -- node[name=u1, pos=0.4] { \(  \mathbf{u}_1\) } (G1.west);
    \draw[->] ([xshift=-2cm]sum.west) -- node[name=ref, pos=0.4] { \(  \mathbf{u}_1 -\check{\mathbf{y}}_2 \) } (sum);
    
    \draw[->]   (G1)--node[name=u1, pos=0.9] { \( \mathbf{y}_1\) }([xshift=2cm]G1.east);
    
    \node [draw, dashed, fit=(G1.east)(G2)(sum)(p), inner sep=0.3cm, label=above:{ \( G_f \) }] {};

    \node[anchor=north east] at ([xshift=0.1cm,yshift=-0.1cm]sum.west) {+};

\end{tikzpicture}

    \caption{
    Positive feedback interconnection given data from individual systems. 
    In general, ${\mathbf{y}}_1\neq \mathbf{u}_2$, so the interconnection is not achievable. In the figure, the successful interconnection is shown with $G_2$ taking ${\mathbf{y}}_1$ as input and generating $\check{\mathbf{y}}_2$ as output. The interconnected feedback system takes the signal from the subtraction $\mathbf{u}_1-\check{\mathbf{y}}_2$ as its input and produces ${\mathbf{y}}_1$ as its output.
    }\vspace{-0.3cm}
    \label{fig:feedback}
\end{figure}
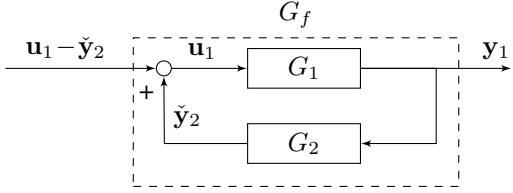

\begin{proof}
 In a positive feedback configuration, the output of $G_1$ becomes the input of $G_2$. Concurrently, $G_1$'s input is derived by subtracting $G_2$'s output from the reference. However, the two trajectories, $\mathbf{w}_1$ and $\mathbf{w}_2$ generally do not satisfy these conditions as they may have been collected in isolation. 
 To set the output of $G_1$ equal the $G_2$'s input, we will feed ${\mathbf{y}}_1$ into $G_2$ to generate $\check{\mathbf{y}}_2 =  \mathcal{Z}({\mathbf{y}}_1,\mathbf{u}_2,\mathbf{y}_2)$. Next, by subtracting $\check{\mathbf{y}}_2$ from $\mathbf{u}_1$, we derive the input to the feedback system, while the output is represented by ${\mathbf{y}}_1$. The trajectory $\mathrm{col}(\mathbf{u}_1-\check{\mathbf{y}}_2, \, {\mathbf{y}}_1)$ would then represent a valid trajectory of the two systems in positive feedback interconnection. An illustration of the proof is given in Fig. \ref{fig:feedback}.
\end{proof}

Note that both Theorems \ref{th: series} and \ref{th: feedback} involve regenerating the output of $G_2$, which requires $\mathbf{w}_2$ to satisfy the low-rank condition. Alternatively, the input and/or output of $G_1$ can be regenerated, in which case $\mathbf{w}_1$ needs to satisfy the low-rank condition instead.
Also, note that Theorems \ref{th: series} and \ref{th: feedback} can also be easily extended to negative feedback, parallel configurations, and more complicated structures. 
In the following section, we will use these results to develop a data-driven parallel to IMC, called IBC.

\section{Internal Behavior Control (IBC)}\label{sec:IBC}
Referring back to  Fig. \ref{fig:IMCBD}, this section introduces the IBC, a data-driven parallel to IMC. We present two IBC approaches: the first is referred to as component-by-component (CBC), where we replace each block in Fig. \ref{fig:IMCBD} with a data-driven predictor. The second is referred to as the unified approach, where we use the results in Section \ref{sec:inter} to combine the data-driven predictors into a single one. 


\subsection{Component-by-Component (CBC) Approach}
In the CBC approach, our objective is to replace the model-based blocks in the IMC structure, specifically $\hat{G}^{-1}$ and $\hat{G}$, with data-driven predictors. For this, we leverage the data-driven forward and inverse predictors, presented in Section \ref{sub:overview BST} and \ref{sub: dd pred}, respectively.



\subsubsection{\textbf{CBC-IBC Forward Model}}
Given a $T_d$-long trajectory $\mathbf{\mathbf{w}_d} = \mathrm{col}(\mathbf{u_d},\mathbf{y_d})$ of  ${{\hat{G}}}$, whose data matrix has rank $T_p+1+n(G)$ for $T_p \geq n(G)$, we can replace the forward model of IMC, i.e., $\hat{G}$, with the data-enabled predictor given in \eqref{eq:g*} and \eqref{eq:ypred}.

\subsubsection{\textbf{CBC-IBC Inverse Model}}
In the original IMC setup, an inverse model of the system, $\hat{G}^{-1}$, is also required to ``cancel out" the system dynamics, effectively making the controlled system (without plant/model mismatch) behave like an identity map from $r$ to $y$. In the CBC-IBC setup, we leverage the inverse predictor reviewed in Section \ref{sub: dd pred} to replace $\hat{G}^{-1}$ with its data-driven counterpart. The inverse predictor uses the same offline data as the forward predictor. {\color{black} While the length of the offline input data remains unchanged, the inverse predictor requires that the output sequence be $L(G)$ steps longer.}  In contrast to the model-based IMC, the inverse predictor in IBC is already causal and implementable because of the delay of $L(G)$ included in the prediction, see Eq.~\eqref{eq:inv sol}, where $u_{\mathrm{pred}}^{\mathrm{inv}}= u(t-L(G))$. 

\subsubsection{\textbf{CBC-IBC Filter}}
In the traditional IMC framework, a low-pass filter with unity DC gain is often employed. A choice  is\footnote{The variable $z$ denotes the argument of the $\mathcal{Z}$-Transform.}:
{\color{black}
\begin{equation}\label{eq:F}
    F(z) = \frac{1}{(\frac{\tau}{T_s} z + (1-\frac{\tau}{T_s}))^{L}},
\end{equation}where $L$ and $T_s$ denote the $L$-delay of system $G$ and the sampling period, respectively.} This choice of order for the filter ensures that the series interconnection of the filter with the system inverse becomes causal and thus implementable. The scalar $\tau \in (0,1)$ is a free parameter, which can be used to trade-off performance and robustness. Smaller $\tau$ values yield faster but less stable responses, while larger values make the response more robust but slower. 

In the CBC-IBC setup, the filter design differs from that in IMC and the unified approach, discussed in the next subsection. This is because a delay of $L({{G}})$ has already been incorporated into the inverse system, as discussed above, and so the data-driven inverse is already implementable. 

Algorithm 1 describes the controller operation in the CBC-IBC at each timestep of the control loop. 
For details on the signal names, please refer to Fig. \ref{fig:IMCBD}. The properties of CBC-IBC are discussed in the following theorem.

\begin{algorithm}
\SetAlgoLined
\let\oldnl\nl
\newcommand{\nonl}{\renewcommand{\nl}{\let\nl\oldnl}}
\nonl \textbf{Setup:} Choose $T_p \geq n(G)$. Make Hankel matrices $U_p, U_f, Y_p, Y_f$, and $Y_{fL}$ through \eqref{eq: hankel} and \eqref{eq:inv} ensuring that both the forward and inverse data matrices have rank $T_p+1+n(G)$. Choose $\tau$ for the filter $F$. \\
\nonl \textbf{External Inputs:} $r(t)$\\
\nonl \textbf{Inputs from the previous loop:} $\mathbf{u}_{[t-T_p,t-1]}, \; \; \; \hat{\mathbf{y}}_{[t-T_p,t-1]}$,
\\ \nonl$\mathbf{s}_{1,[t-T_p-L(G),t-1]},\; \; \; \mathbf{s}_{2,[t-T_p-L(G),t-L(G)-1]}$ \\
\nonl \textbf{Outputs:} Control command $u(t)$\\
Compute $e(t) = y(t) - \hat{y}(t-1)$.\\
Compute $s_1(t) = r(t) - e(t)$.\\
    \text{Data-enabled prediction of System Inverse}: $s_{2}(t) = U_f\begin{bmatrix}
        Y_p \\ U_p \\ Y_{fL}
    \end{bmatrix} ^\dagger \begin{bmatrix}
        \mathbf{s}_{1,[t-T_p-L(G),t-L(G)-1]}\\\mathbf{s}_{2,[t-T_p-L(G),t-L(G)-1]}\\\mathbf{s}_{1,[t-L(G),t]}
    \end{bmatrix}$.\\
    Compute, $u(t)$, by filtering $s_{2}(t)$ through $F$.\\
    \text{Compute the data-enabled prediction (for the next  loop): }\\
    \nonl$\hat{y}(t) = Y_f \begin{bmatrix}
        U_p \\ Y_p \\ U_{f}
    \end{bmatrix}^\dagger \begin{bmatrix}
        \mathbf{u}_{[t-T_p,t-1]}\\ \hat{\mathbf{y}}_{[t-T_p,t-1]}\\u(t)
    \end{bmatrix}$.\\
\caption{CBC-IBC controller implementation}\label{algor_1}
\end{algorithm}

\begin{theorem}\label{th: ident}
    Suppose $G$ is an open-loop-stable minimum-phase SISO LTI system. The CBC-IBC controller using filter $z^LF(z)$ performs identically to the classical IMC controller using filter $F(z)$. Furthermore, the IBC controller results in perfect steady-state tracking and disturbance rejection, and minimizes the $H_2$ norm of the tracking error for step commands.
\end{theorem}
\begin{proof}
    We know that the model-based forward predictions in IMC are identical to their data-driven counterparts in IBC. For the inverse prediction, the two differ by $L(G)$ delays as discussed earlier. The statement of the theorem thus follows. Finally, the tracking and disturbance rejection properties and $H_2$ optimality are directly inherited from the model-based IMC.
\end{proof}

\subsection{Unified Approach}
It is known that the IMC controller can be represented as a single LTI controller, ${C}$,  as shown in Fig.~\ref{fig:enter-label}. This controller is given by $C(s)=\frac{\hat{G}^{-1}F}{1-F}$, which has the same order as $G$, i.e., $n(C)=n(G)$. In this section, we apply the tools developed in section \ref{sec:inter} to create a unified IBC approach that uses a single data-driven predictor to replace $C$, thus eliminating the necessity of separate data-enabled predictors for the forward and inverse systems, as was required in the CBC approach. The unified approach thus yields a more compact controller compared to CBC-IBC.

\begin{figure}
    \centering
    \begin{tikzpicture}[auto, node distance=2cm,>=latex',
  sum/.style={
    draw, 
    circle, 
    inner sep=2pt
  }]
    \node [draw, rectangle, minimum width=1.3cm] (G1) { \( \hat{G}^{-1} \) };
    \node [draw, rectangle, right of=G1 , node distance=2cm, minimum width=1.3cm] (G2) { \( F \) }; 
    \node [draw, rectangle, right of=G2 , node distance=2cm, minimum width=1.3cm] (G3) { \( G \) }; 
    \node [draw, rectangle, below right of=G1 , node distance=1.5cm, minimum width=1.3cm] (G4) { \( \hat{G} \) }; 
    
    \node[sum,  left=0.5cm of G1] (sum1) {};
    \node[sum,  left=0.5cm of sum1] (sum2) {};

    \draw[->]   ([xshift=-1cm]sum2.west)--node[name=s3, midway] { \( r \) }(sum2);
    \draw[->] (G3) -- node[name=y, midway] { \( y \) } ([xshift=1cm]G3.east);

    \draw[->] (G1.east) -- (G2.west);
    \draw[->] (sum1.east) -- (G1.west);
    \draw[->] (sum2.east) -- node[name=s3, pos=0.09] { \( s_3 \) }(sum1.west);
    \draw[->] (G2.east) -- node[name=u, pos=0.2] {\( u\)}(G3.west);
    \draw[->] (G4.west)-- ++(-1.65cm,0) -- (sum1.south);
    \draw[->] (G2.east)-- ++(0.2cm,0)-- ++(0,-1.05cm) -- (G4.east);
    \draw[->] (G3.east)-- ++(0.5cm,0)-- ++(0,-2cm) -- ++(-7.15cm,0)--(sum2.south);

    \node [draw, dashed, fit=(G1)(G2)(G4)(sum1)(u), inner sep=0.225cm, label=above:{ \( C \) }] {};

    \node[anchor=north east] at ([xshift=0.1cm,yshift=-0.1cm]sum1.west) {+};

    \node[anchor=north east] at ([xshift=0.1cm,yshift=-0.1cm]sum2.west) {-};

\end{tikzpicture}

    \caption{Simplified block diagram of the IMC structure}
    \label{fig:enter-label}
\end{figure}
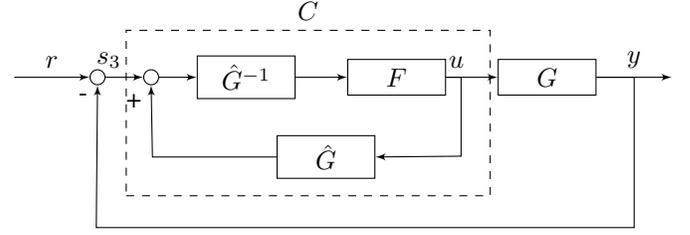

We now apply the results of Section \ref{sec:inter} to create a data-driven representation of $C$, which we achieve by finding an admissible trajectory for $C$ from the input-output data of $G$. The results are summarized in the following theorem.

\begin{theorem} 
Given
\begin{itemize}
    \item $T_d$-long input-output data $\mathbf{u}^d$ and $\mathbf{y}^d$ from the plant, $G$, where the data matrix has rank $T_p+1+n(G)$ for $T_p \geq n(G)$,  
    \item and filter $F$ in \eqref{eq:F},
\end{itemize} 
then the trajectory defined by
\begin{equation*}
    \mathbf{w}_c  = \mathrm{col}({\mathbf{y}}^d-{{\bar{\mathbf{y}}}^d},  {\bar{\mathbf{u}}^d}), \quad \, \bar{\mathbf{y}}^d = f*{\mathbf{y}}^d, \quad \, \bar{\mathbf{u}}^d = f*\mathbf{u}^d ,
\end{equation*}
satisfies $\mathbf{w}_c \in \mathscr{B}|_{T_d}(C)$, where $f$ is the impulse response of the filter $F$, and `*' denotes the convolution operator.
\end{theorem}

\begin{proof}
    The proof directly stems from Theorems \ref{th: series} and \ref{th: feedback} with the input-output sequences for the inverse system set to $\mathbf{y}^d$ and $\mathbf{u}^d$.
\end{proof}

Now, if the data matrix of $\mathbf{w}_c$ has rank $T_p+1+n(C)$ for $T_p \geq n(C)$, then this data matrix can be used in the IBC unified approach to make data-enabled predictions of $C$. Note that $n(C) = n(G)$. Specifically, the data matrix of $\mathbf{w}_c$ is defined as,

{\color{black}\begin{equation}\label{eq: Hc}
    \mathcal{H}_C = \begin{bmatrix}
    \mathcal{H}_{C,u} \\ \mathcal{H}_{C,y}
\end{bmatrix}\in \mathbb{R}^{2(T_p+1)\times (T_d-T_p)}
\end{equation} where 
\begin{equation*}
    \mathcal{H}_{C,u} = \mathscr{H}_{T_p+1}(\textbf{y}^d-\bar{\textbf{y}}^d) =\begin{bmatrix}
    E_p \\ E_f
\end{bmatrix}, \, \mathcal{H}_{C,y} = \mathscr{H}_{T_p+1}(\bar{\textbf{u}}^d) = \begin{bmatrix}
    F_p \\ F_f
\end{bmatrix},\end{equation*} with the partitioning $E_p, F_p \in \mathbb{R}^{T_p\times(T_d-T_p)}$ and $E_f, F_f \in \mathbb{R}^{1\times(T_d-T_p)}$. 
} Data-enabled predictions based on $\mathcal{H}_C$ will be used directly in the unified IBC approach. {\color{black}Theorem \ref{th: ident} also applies to the unified approach, guaranteeing an identical performance to the classical IMC controller using the same filter.} To conclude this section, Algorithm 2 presents the unified IBC controller at each timestep of the control loop. For details of the signal names, please see Fig. \ref{fig:enter-label}.

\begin{algorithm}
\SetAlgoLined
\let\oldnl\nl
\newcommand{\nonl}{\renewcommand{\nl}{\let\nl\oldnl}}
\nonl \textbf{External Inputs:} $r(t)$\\
\nonl \textbf{Inputs from the previous loop:} $\mathbf{u}_{[t-T_p,t-1]},\; \; \;\mathbf{s}_{3,[t-T_p,t-1]}$ \\
\nonl \textbf{Outputs:} Control command $u(t)$\\
\nonl \textbf{Setup:} Set $T_p \geq n(G)$. Filter the input-output sequence $\mathbf{u}^d$ and ${\mathbf{y}}^d$  through $F$ to obtain $\bar{\mathbf{{u}}}_d$ and $\bar{\mathbf{{y}}}_d$, respectively. Then, using the difference ${\mathbf{y}}^d-\bar{\mathbf{y}}^d$ as input and $\bar{\mathbf{{u}}}_d$ as output, generate the Hankel matrices $E_p, E_f, F_p,$ and $F_f$ through Eq. \eqref{eq: Hc}, ensuring that the resulting data matrix has rank $T_p+1+n(G)$.\\
Compute $s_3(t) = r(t) - y(t)$.\\
    \text{Compute and return}:
    $u(t) = F_f\begin{bmatrix}
        E_p \\ F_p \\ E_{f}
    \end{bmatrix}^\dagger \begin{bmatrix}
        \mathbf{s}_{3,[t-T_p,t-1]}\\ \mathbf{u}_{[t-T_p,t-1]}\\s_3(t)
    \end{bmatrix}$. \\
    
    \caption{Unified-IBC controller implementation}\label{algor_2}
\end{algorithm}

\begin{remark}
    We recall from \cite{garcia1982internal} that IMC and constraint-free MPC are essentially identical in the sense that for each tuning of the MPC controller, there exists a unique IMC filter, $F$, that results in the same control actions. A similar relationship exists between their data-driven counterparts. Specifically, just as IBC serves as a data-enabled parallel of IMC, DeePC \cite{coulson2019data} is a data-driven parallel of MPC. Therefore, it is reasonable to infer that the same essential equivalence exists between IBC and DeePC when no constraints are involved and a similar performance is to be expected from both methods in the noise-free setup.

\end{remark}

\begin{remark}\label{rem:minimum}
The unified approach has less memory requirements than CBC, in terms of both online and offline data. Regarding online data, CBC requires solving two data-enabled predictions, one for the forward and one for the inverse model. {\color{black} The forward prediction requires a minimum memory of the past $n(G)$ input-output samples, while the inverse requires the past $n(G)+L({G})$ samples.} The unified approach, on the other hand, solves for only one data-enabled prediction in every control loop but requires memory of the past $n(G)$ samples. Therefore, the unified requires $n(G)+L(G)$ fewer samples and is thus more memory efficient in terms of online data. Regarding offline data, CBC requires a minimum of $T_d = 3n+1+L$ to satisfy the low-rank condition, while unified requires $T_d = 3n+1$, so the unified approach requires $L$ fewer offline data samples.
\end{remark}
\section{Simulation Results}
In this section, we illustrate the proposed IBC controller using an example. Consider the second-order plant,
$
    G(s) = \frac{10(s+1)}{(s+2)(s+4)}.
$
The dynamics are discretized using the zero-order hold method with a sampling period of $0.01$ s. The time constant, $\tau$, of the filter  is set to $0.5$ s. {\color{black}In the CBC approach, the filter is advanced by a time shift of $L=1$ to avoid unnecessary delays (consistent with Theorem \ref{th: ident})}. For the IBC controller design, we collect offline data by feeding the system a random input sequence. {\color{black}The collected output, $\mathbf{y}^d$, as well as the filtered input and output signals used in the unified approach, are displayed in Fig.~\ref{fig:illust_pre}. Since $n(G)=2$, $T_p$ is chosen as 2 in both approaches. The number of offline data samples, $T_d$, is chosen as the minimum: $T_d=8$ for CBC, and $T_d=7$ (see Remark \ref{rem:minimum}).}
 Both algorithms achieve perfect tracking as depicted in Fig. \ref{fig:illust_result}. {\color{black} A step disturbance is applied to the plant input around time $13$ s. Both IBC approaches successfully recover from this disturbance.} Note that while both approaches yield identical performance, the unified approach is more data efficient and solves a single data-enabled prediction in each iteration, in contrast to the CBC approach, which solves two. 


\begin{figure}
    \centering
    \includegraphics[trim=0cm 0cm 1.1cm 0.6cm, clip, width = 0.7\columnwidth]{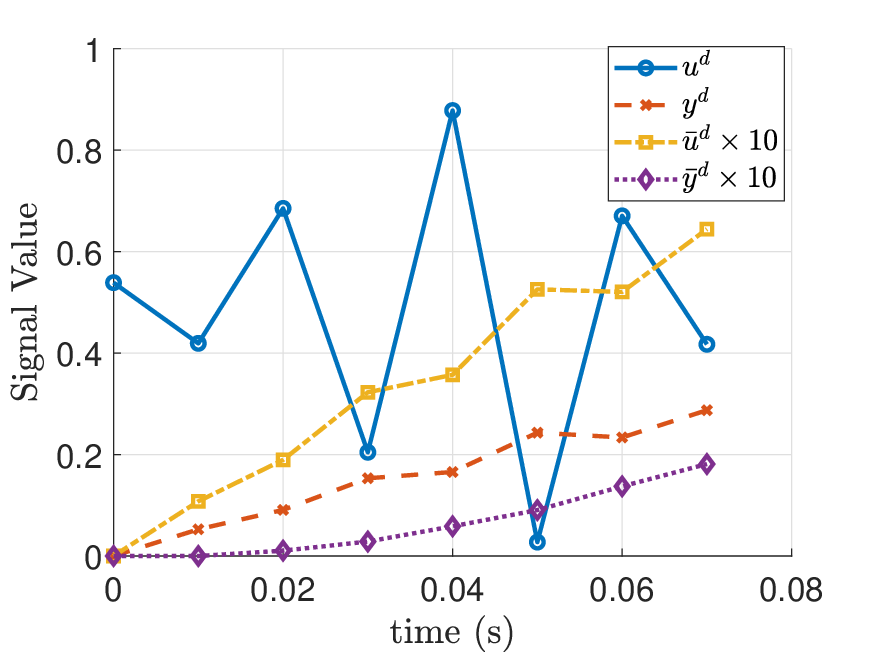}
    \caption{Noise-free raw and filtered input/output data collected offline.{\color{black} Because the input is random, all Hankel matrices satisfy their respective low-rank conditions.}}
    \label{fig:illust_pre}
\end{figure}

\begin{figure}
    \centering
    \includegraphics[trim=0cm 0cm 1.1cm 0.6cm, clip, width = 0.7\columnwidth]{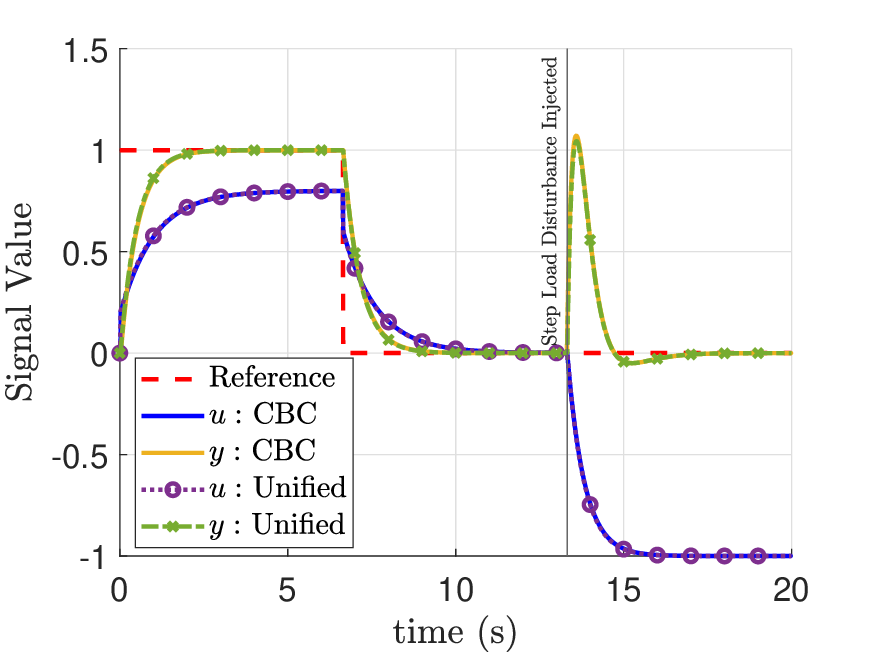}
    \caption{Illustration of the effectiveness of IBC in both CBC and unified approaches in reference tracking and disturbance rejection.}
    \label{fig:illust_result}
\end{figure}

\section{Conclusions and Future Work}

This paper first presented a data-driven formulation of series and feedback interconnection of dynamical systems, where we described how data-driven predictors can be designed for the interconnected system using data from each individual system. 
Using these ideas, we introduced the Internal Behavior Control (IBC), a new data-driven control strategy based on the well-known Internal Model Control (IMC). 
Similar to IMC, the IBC is easy to tune and results in perfect tracking and disturbance rejection but, unlike IMC, does not require a parametric model. We outlined two approaches for IBC implementation: CBC and Unified. We comapred the two approaches in terms of filter design, computations, and memory requirements, where we showed that unified is more memory efficient. 


There are numerous avenues for future research. First, the focus of this paper was on open-loop stable, minimum-phase SISO LTI systems. Future work will relax these assumptions to study unstable, non-minimum phase, and/or MIMO systems in the context of IBC. Second, we assumed here that the offline data is free of ``noise" - measurement noise, process disturbance, or nonlinearities, all of which make the data inconsistent with an LTI structure. Investigation of noise in the offline data and its impact on the IBC performance is another topic of future work. 
Third, in traditional IMC, the filter is tuned using the knowledge of process uncertainties. Future work will investigate uncertainty quantification of the data-driven predictors and IBC filter tuning. 
Finally, we did not offer specific guidelines on how to design the offline data to ensure that the data matrices for the series and feedback interconnections satisfy their respective rank conditions. Input design is another interesting topic for future work. 



\bibliographystyle{unsrt}
\bibliography{bib}

\begin{thebibliography}{10}

\bibitem{willems1997introduction}
Jan~C Willems and Jan~W Polderman.
\newblock {\em Introduction to mathematical systems theory: a behavioral approach}, volume~26.
\newblock Springer Science \& Business Media, 1997.

\bibitem{willems2007behavioral}
Jan~C Willems.
\newblock The behavioral approach to open and interconnected systems.
\newblock {\em IEEE control systems magazine}, 27(6):46--99, 2007.

\bibitem{willems1991paradigms}
Jan~C Willems.
\newblock Paradigms and puzzles in the theory of dynamical systems.
\newblock {\em IEEE Transactions on automatic control}, 36(3):259--294, 1991.

\bibitem{willems1989models}
Jan~C Willems.
\newblock Models for dynamics.
\newblock {\em Dynamics reported: a series in dynamical systems and their applications}, pages 171--269, 1989.

\bibitem{maupong2017data}
Thabiso~M Maupong and Paolo Rapisarda.
\newblock Data-driven control: A behavioral approach.
\newblock {\em Systems \& Control Letters}, 101:37--43, 2017.

\bibitem{maupong2016data}
Thabiso Maupong and Paolo Rapisarda.
\newblock Data-driven control: the full interconnection case.
\newblock In {\em 22nd International Symposium on Mathematical Theory of Networks and Systems}, July 2016.

\bibitem{trentelman1996control}
Harry~L Trentelman and Jan~C Willems.
\newblock Control in a behavioral setting.
\newblock In {\em Proceedings of 35th IEEE Conference on Decision and Control}, volume~2, pages 1824--1829. IEEE, 1996.

\bibitem{willems1997interconnections}
Jan~C Willems.
\newblock On interconnections, control, and feedback.
\newblock {\em IEEE Transactions on Automatic control}, 42(3):326--339, 1997.

\bibitem{garcia1982internal}
Carlos~E Garcia and Manfred Morari.
\newblock Internal model control. a unifying review and some new results.
\newblock {\em Industrial \& Engineering Chemistry Process Design and Development}, 21(2):308--323, 1982.

\bibitem{coulson2019data}
Jeremy Coulson, John Lygeros, and Florian D{\"o}rfler.
\newblock Data-enabled predictive control: In the shallows of the deepc.
\newblock In {\em 2019 18th European Control Conference (ECC)}, pages 307--312. IEEE, 2019.

\bibitem{favoreel1999model}
Wouter Favoreel.
\newblock {\em Subspace methods for identification and control of linear and bilinear systems}.
\newblock PhD thesis, Katholiek Universiteit Leuven, 1999.

\bibitem{rueda2020data}
Juan~G Rueda-Escobedo and Johannes Schiffer.
\newblock Data-driven internal model control of second-order discrete volterra systems.
\newblock In {\em 59th IEEE Conference on Decision and Control (CDC)}, pages 4572--4579. IEEE, 2020.

\bibitem{MARKOVSKY202142}
Ivan Markovsky and Florian Dörfler.
\newblock Behavioral systems theory in data-driven analysis, signal processing, and control.
\newblock {\em Annual Reviews in Control}, 52:42--64, 2021.

\bibitem{markovsky2022identifiability}
Ivan Markovsky and Florian D{\"o}rfler.
\newblock Identifiability in the behavioral setting.
\newblock {\em IEEE Transactions on Automatic Control}, 68(3):1667--1677, 2022.

\bibitem{dorfler2022bridging}
Florian D{\"o}rfler, Jeremy Coulson, and Ivan Markovsky.
\newblock Bridging direct and indirect data-driven control formulations via regularizations and relaxations.
\newblock {\em IEEE Transactions on Automatic Control}, 68(2):883--897, 2022.

\bibitem{eun2023data}
Yongsoon Eun, Jaeho Lee, and Hyungbo Shim.
\newblock Data-driven inverse of linear systems and application to disturbance observers.
\newblock In {\em 2023 American Control Conference (ACC)}, pages 2806--2811. IEEE, 2023.

\bibitem{morari1987robust}
Manfred Morari.
\newblock Robust process control.
\newblock {\em Chemical engineering research \& design}, 65(6):462--479, 1987.

\bibitem{rivera1986internal}
Daniel~E Rivera, Manfred Morari, and Sigurd Skogestad.
\newblock Internal model control: Pid controller design.
\newblock {\em Industrial \& engineering chemistry process design and development}, 25(1):252--265, 1986.

\bibitem{aguirre2005interpretation}
LA~Aguirre, MCS Coelho, and MV~Correa.
\newblock On the interpretation and practice of dynamical differences between hammerstein and wiener models.
\newblock {\em IEE Proceedings-Control Theory and Applications}, 152(4):349--356, 2005.

\end{thebibliography}

\end{document}